\documentclass[11pt,a4paper]{scrartcl}
\usepackage[margin=0.9in,bottom=1.2in]{geometry}
\usepackage{amsfonts,amssymb,amsmath}
\usepackage[thmmarks,hyperref,amsthm,amsmath]{ntheorem} 
\usepackage{graphicx}
\usepackage[usenames,dvipsnames]{color}
\usepackage{hyperref}
\usepackage[utf8]{inputenc}
\usepackage{pxfonts}
\usepackage[labelfont=bf,font=small]{caption}
\usepackage{authblk}
\usepackage{xspace}

\definecolor{blueLink}{rgb}{0,0.2,0.8}
\hypersetup{colorlinks,linkcolor=blueLink,urlcolor=blueLink,citecolor=blueLink}
\newcommand{\lref}[2][]{\hyperref[#2]{#1~\ref*{#2}}}

\newtheorem{theorem}{Theorem}
\newtheorem{lemma}[theorem]{Lemma}

\newcommand{\OPT}{\textsc{Opt}\xspace}

\newcommand{\BT}{\textsc{BT}\xspace}
\newcommand{\eps}{\varepsilon}

\title{Logarithmic price of buffer downscaling \\ on line metrics\thanks{%
This work was partially supported by Polish National Science Centre grant 2016/22/E/ST6/00499,
by the Center of Excellence -- ITI, project P202/12/G061 of GA~\v{C}R (J. Sgall), and by the
project 634217 of GAUK (M. B\"ohm, P. Vesel\'y).}}

\author[1]{Marcin Bienkowski}
\author[2]{Martin B\"{o}hm}
\author[1]{{\L}ukasz Je\.{z}}
\author[1,3]{Paweł Laskoś-Grabowski}
\author[1]{Jan Marcinkowski}
\author[2]{Ji\v{r}\'{\i} Sgall}
\author[1]{Aleksandra Spyra}
\author[2]{Pavel Vesel\'{y}}
\affil[1]{Institute of Computer Science, University of Wrocław, Poland}
\affil[2]{Computer Science Institute, Charles University, Czech Republic}
\affil[3]{Electron (Chaddenwych Services Ltd.), London, UK}
\date{}

\begin{document}

\maketitle

\vspace{-1cm}

\begin{abstract}
We consider the reordering buffer problem on a  line consisting of $n$
equidistant points.  We show that, for any constant $\delta$,  an (offline)
algorithm that has a buffer $(1-\delta) \cdot k$ performs worse by a factor of 
$\Omega(\log n)$ than an offline algorithm with buffer~$k$. In particular, this
demonstrates that the $O(\log n)$-competitive online algorithm
\textsc{MovingPartition} by Gamzu and Segev (ACM Trans.\ on Algorithms, 6(1),
2009) is essentially optimal against any offline algorithm with a
slightly larger buffer. 
\end{abstract}


\section{Introduction}

In the reordering buffer problem, requests arrive in an online fashion at the
points of a metric space and have to be served by an algorithm. An algorithm
has a single server kept at a point of the metric space and is equipped
with a finite buffer, whose capacity is denoted by $k$. The buffer
is used to give the algorithm a~possibility of serving requests in 
a~different order. Namely, at any time, there can be at most $k$ unprocessed
requests, and once their number is exactly~$k$, the algorithm has to move its
server to the position of a~pending request of its choice. The request is
then considered served and removed from the buffer. The goal is to minimize 
the total distance traveled by the server. 

The problem was coined by R{\"{a}}cke et al.~\cite{reordering-buffer-first}
and the currently best algorithm for general metric spaces is a randomized
$O(\log n \cdot \log k)$-competitive strategy~\cite{reordering-buffer-hst,reordering-buffer-hst-better},
where $n$ is the number of points in a metric space. Most of the research
focused however on specific metric spaces, such as uniform
metrics~\cite{reordering-buffer-first,reordering-buffer-randomized-rabani,reordering-buffer-map},
stars~\cite{reordering-buffer-adamaszek,reordering-buffer-randomized-stars,reordering-buffer-rabani-soda}
or lines of $n$ equidistant
points~\cite{reordering-buffer-gamzu-segev,reordering-buffer-line-first}.

\subsection{Two-stage approach}
One of the straightforward ways to attack the problem is by a two-stage
approach: (i) compare an~online algorithm to an optimal offline algorithm that 
has a smaller buffer $h < k$ and (ii)~bound the ratio between
optimal offline algorithms equipped with buffers $h$ and $k$. 
Such an approach was successfully executed by Englert and
Westermann~\cite{reordering-buffer-map} for uniform metrics. They constructed an
online algorithm \textsc{Map} (with a buffer $k$) that is $4$-competitive
against $\OPT(k/4)$,
where $\OPT(s)$ denotes the optimal offline algorithm with buffer~$s$.
Subsequently, they showed that on any instance the
costs of $\OPT(k/4)$ and $\OPT(k)$ can differ by a factor of at most $O(\log k)$. 
This implies that \textsc{Map} is $O(4 \cdot \log k) = O(\log k)$-competitive. 

As shown by Aboud~\cite{reordering-buffer-aboud}, for the uniform metric space
the relation between $\OPT(k/4)$ and $\OPT(k)$ cannot be asymptotically
improved, which excludes the option of beating the ratio $O(\log k)$ by the
described two-stage process. Note that this does not rule out the possibility
of decreasing the ratio by another approach. Indeed, subsequent works gave
improved results, essentially resolving the uniform metric case for
deterministic algorithms: Adamaszek et al.~\cite{reordering-buffer-adamaszek}
gave an $O(\!\sqrt{\log k})$-competitive deterministic strategy and showed
that the competitive ratio of every deterministic algorithm is
$\Omega(\!\sqrt{\log k / \log \log k})$.

\subsection{Two-stage approach for line metrics (our result)}

Arguably, the most uncharted territory is the line metric, or more
specifically, a line graph consisting  of $n$ equidistant sites. There, 
the lower bound on the competitive ratio is only~$2.154$~\cite{reordering-buffer-gamzu-segev}.
On the other hand, the best known strategy is an $O(\log n)$-competitive algorithm
\textsc{MovingPartition} by Gamzu and Segev~\cite{reordering-buffer-gamzu-segev}.
Note that achieving an upper bound of $O(k)$ is possible for any metric 
space~\cite{reordering-buffer-map}, and hence the competitive ratio for 
line metrics is $O(\min \{k, \log n\})$. 
In this paper, we show that it is not possible to improve this bound 
by the two-stage approach described above, by proving the following result.

\begin{theorem}
\label{thm:main_theorem}
Fix a line consisting of $n$ equidistant sites, any $k$, and any constant 
$\delta \in (0,1)$. There exists an~input sequence on which the ratio between the costs of 
$\OPT(k)$ and $\OPT((1-\delta) \cdot k)$ is at least 
$\Omega(\min \{ k, \log n\})$.
\end{theorem}

Our result has additional consequences for the two known online 
algorithms for the reordering buffer problem: 
\textsc{MovingPartition}~\cite{reordering-buffer-gamzu-segev} and 
\textsc{Pay}~\cite{reordering-buffer-hst-better}.

A straightforward modification of the analysis
in~\cite{reordering-buffer-gamzu-segev} shows that
\textsc{MovingPartition} is in fact $O((h/k) \cdot \log n)$-competitive against
an~optimal offline algorithm that has a larger buffer $h > k$. 
Our result implies that, assuming $\log n = O(k)$, 
for $h = k \cdot (1+\eps)$ and a~fixed $\eps > 0$,
this ratio is asymptotically optimal.

Similarly, algorithm \textsc{Pay}~\cite{reordering-buffer-hst-better} achieves a
competitive ratio of $O((h/k) \cdot (\log D + \log k))$ on trees with
hop-diameter $D$, and hence it is $O((h/k) \cdot (\log n + \log k))$-competitive
when applied to a line metric.
Therefore, \textsc{Pay} is asymptotically optimal for line metrics when 
when $n = O(k)$ and $h = k \cdot (1+\eps)$ for a~fixed $\eps > 0$.


\section{Main construction}

The main technical contribution of this paper is to show
\lref[Theorem]{thm:main_theorem} for the case when the line consists of exactly 
$n = 2^k + 1$ sites. For such a setting, the key lemma given below yields
the cost separation of $\Omega(k) = \Omega(\log n)$. In the next section, we
show how request grouping extends this result from the case of $n = 2^k + 1$
sites to arbitrary values of $k$ and $n$. 

\begin{lemma}
\label{lem:main_lemma}
Fix two integers $0 < \ell' < \ell$ and a constant $\eps$ such that 
$\ell' = (1-\eps) \cdot \ell$. Assume that the line consists
of $2^\ell+1$ equidistant sites. There exists an input sequence on which
the ratio between the costs of $\OPT(\ell)$ and $\OPT(\ell')$ is
at least $\ell \cdot \eps^2 \cdot (1+\eps)^{-1} \cdot \log_2^{-1} (1+1/\eps)$.
\end{lemma}

We number the sites from $0$ to $2^\ell$. For simplicity of the description,
we associate times with request arrivals: all requests arrive only during a
\emph{step} between two consecutive integer times. For a~single step, we will
explicitly specify the positions of the requests arriving in this step but not
their time ordering, which is irrelevant. (If needed, one may assume that they
arrive in the increasing site numbering order.)
The input sequence consists of \emph{phases}, one of which we describe
below. Then, any subsequent phase is ``mirrored'' with respect to the previous
one, i.e., the description is identical except for reversing the site
numbering.

The phase consists of $2^\ell$ steps and it begins at time 0. Each request
is either \emph{regular} or \emph{auxiliary}. Each regular request has a
\emph{rank} $i$ from $0$ to $\ell-1$ and their positions are defined 
by the recursive construction of \emph{blocks} described below.

\begin{figure}[t]
\centering
\includegraphics[width=0.66\textwidth]{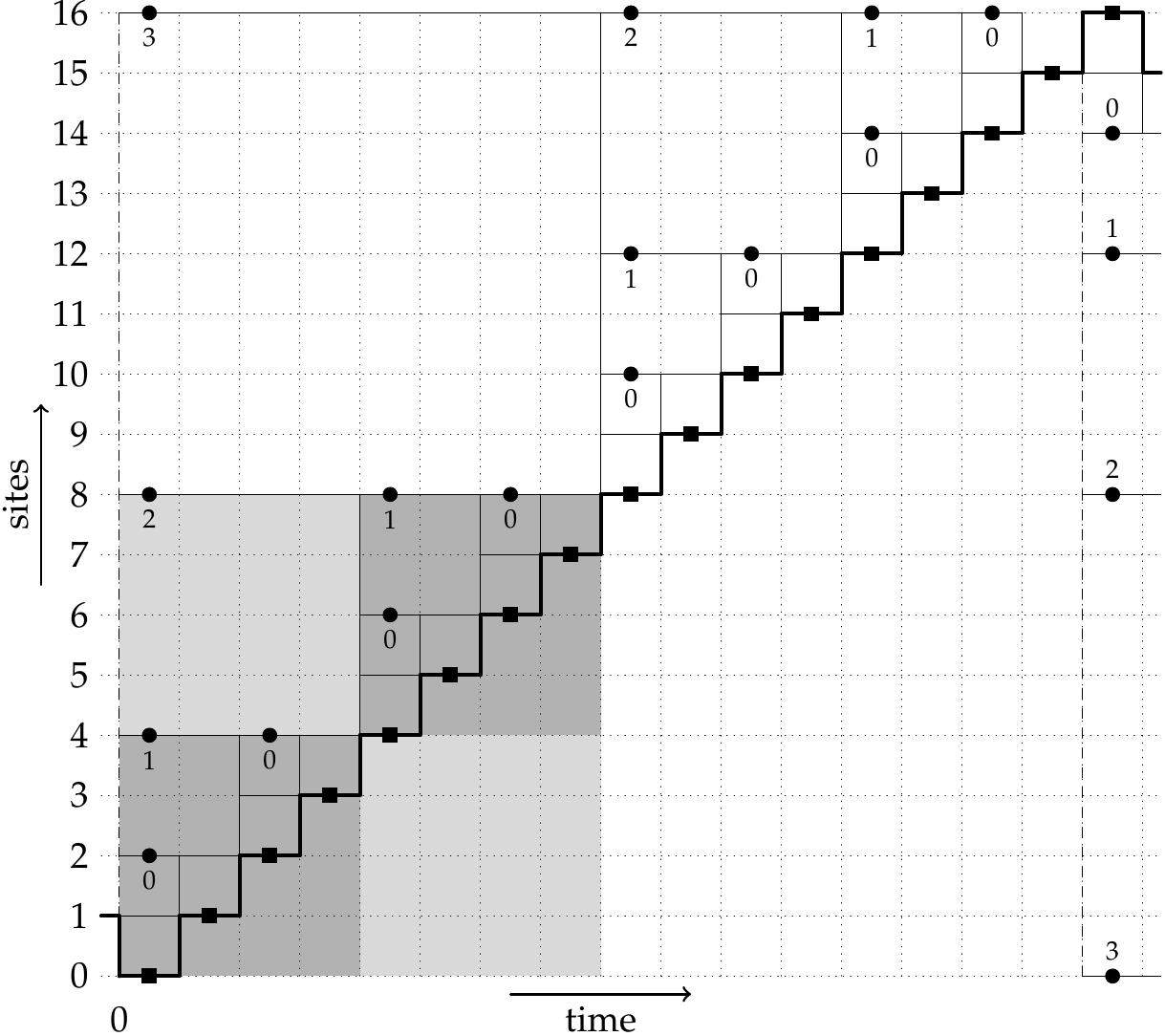}
\caption[One phase of the input sequence]{
A phase of the input sequence for $\ell=4$ (or, a line of 17 sites).
Horizontal dotted lines correspond to sites, whereas vertical dotted lines
correspond to integer times.  The dashed lines indicate the beginning and
the end of the phase.
Disks and squares denote respectively the regular requests (annotated
with their ranks) and the anchors (sets of $\ell + 1$ auxiliary requests each).
A block of rank $3$ and its two sub-blocks of rank $2$ are marked at the bottom left;
the block contains one regular request of rank $2$.
The thick line represents the basic trajectory.
The thinly-bordered squares have 
sides of length equal to powers of 2 and are drawn purely for clarity of the
input construction.}
\label{fig:input}
\end{figure}

For any integers $q \in \{ 1, \ldots, \ell\}$ and 
$s \in \{0, \ldots, 2^{\ell-q}-1\}$, a \emph{$(q,s)$-block} is
a square portion of space-time, consisting of all times in $(2^q \cdot s,
2^q \cdot (s+1)]$ and all sites in $(2^q \cdot s,\,2^q \cdot (s+1)]$,
where $q$ is called the
\emph{rank} of the block. 
For instance, the $(\ell,0)$-block is the whole phase. Moreover, 
$(q,s)$-block contains two \emph{sub-blocks}: the $(q-1,\,2s)$-block and 
the $(q-1,\,2s+1)$-block, see~\lref[Figure]{fig:input}.
In each $(q,s)$-block a unique request of rank $q-1$ arrives at site 
$2^q \cdot (s+1)$ in the step following time $2^q \cdot s$;
note that this request is not contained in any of its two sub-blocks. 
Graphically, it is near the top left corner of the block, 
and it is depicted as a disk in~\lref[Figure]{fig:input}.

In turn, auxiliary requests are grouped in sets of $\ell + 1$ called
\emph{anchors} (squares in \lref[Figure]{fig:input}). For all $j$ from $0$ to $2^\ell-1$, the
$j$-th anchor arrives at site $j$ in the~step following time $j$.

Throughout the paper, any relative directions in the text pertain to this figure, 
e.g., ``up'' and ``right'' mean ``with growing site numbers'' and ``forward in time'', 
respectively.


\subsection{Upper bound on the cost of the optimal algorithm with a larger buffer}

Let the \emph{basic trajectory} (depicted as a thick line
in~\lref[Figure]{fig:input}) be the movement of a server that, for all~$j$ from
0 to $2^\ell-1$, remains at site $j$ throughout the step between the integer
times $j$ and $j+1$, and moves directly from $j$ to $j+1$ at the integer time
$j+1$.  Note that in the subsequent phase the last site becomes site $0$, and
thus the basic trajectory remains continuous over multiple phases as well.

\begin{lemma}
The cost of $\OPT(\ell)$ in a single phase is at most $2^\ell$.
\label{lem:OPTcost}\end{lemma}

\begin{proof}
We prove that an algorithm with buffer size $\ell$ that follows the basic 
trajectory, denoted $\BT(\ell)$, yields a feasible solution. As its cost is
clearly $2^\ell$, the lemma follows.  First, all anchors are served 
immediately by $\BT(\ell)$, without storing them in the buffer, as
they arrive at a current position of the server. Second, we will show that for each of $\ell$
possible ranks, $\BT(\ell)$ has at most one regular request of the given rank
in its buffer, and hence its buffer capacity is not exceeded.
To this end, observe that the $j$-th request of rank $i$ arrives at site 
$2^{i+1} \cdot j$ in a step following time $2^{i+1} \cdot (j-1)$. It is served 
at time $2^{i+1} \cdot j$ by $\BT(\ell)$, i.e., right before the next request 
of this rank arrives.

In fact, the cost of $\BT(\ell)$ is optimal as the anchors placed at all
sites force any algorithm with buffer $\ell$ to traverse the whole line at
least once per phase.
\end{proof}


\subsection{Lower bound on the cost of the optimal algorithm with a smaller buffer}

In this part, we construct a lower bound that holds for any algorithm that has
buffer $\ell$ or smaller. Later, we apply this bound to $\OPT(\ell')$; recall
that $\ell' = (1-\eps) \cdot \ell$. Observe that any such algorithm has to
process each anchor as soon as it arrives, since the anchor has $\ell+1$
requests. Therefore, its trajectory within the step between times $j$ and $j+1$
has to pass through site $j$. Informally speaking, its trajectory has to touch
or intersect the basic trajectory in all steps of the phase.

Regular requests inside a $(q,s)$-block are called \emph{new} for this block, whereas
regular requests that arrived at any site in $(2^q \cdot s,\,2^q \cdot (s+1)]$
before time $2^q \cdot s$ are called \emph{old} for this block.
Note that all the old requests for the $(q,s)$-block are at a single site, namely at 
$2^q \cdot (s+1)$.
Recall that $(q,s)$-block contains only one request not contained in its sub-blocks
(with rank $q-1$ and position at site $2^q \cdot (s+1)$).

For $p,r \geq 0, q \geq 1$, let $T(p,q,r)$ be the minimal cost of any
algorithm for serving a block of rank~$q$, assuming that at the beginning of
the block, its buffer has space for $p$ (new) requests and $r$ old requests 
\emph{for this block} are already stored in the buffer.
If the algorithm trajectory leaves the block, $T(p,q,r)$ accounts only for these parts of its
trajectory that are contained completely inside the block. We show the
following lower bound on $T(p,q,r)$.

\begin{lemma}
\label{lem:T_relation}
$T(p,q,r) \geq 2^q-1$. Furthermore, for $q \geq 2$,
\begin{align*}
  T(p,q,r) & \geq \min\,\left\{\,
    2^q+2\cdot T(p+r,q-1,0)\; , \;
    T(p-1,q-1,0)+T(p-1,q-1,r+1)\,
  \right\} 
  & \quad \text{if $p \geq 1$ \enspace,} \\
  T(p,q,r) & \geq 
    \phantom{\min \;\;\;} 2^q+2\cdot T(p+r,q-1,0)
  & \quad \text{if $p = 0$ \enspace.}
\end{align*}
\end{lemma}

\begin{proof}
The first inequality holds, because a block of rank $q$ contains anchors at all 
its sites, and the distance between the two most distant of those is $2^q-1$.
For the remaining two inequalities (for $q \geq 2$), we consider two possible 
behaviors of the algorithm.
\begin{enumerate}
\item The algorithm moves the server to site $2^q \cdot (s+1)$ before time
$2^q \cdot s + 2^{q-1}$. That is, it moves the server to the upper boundary of
the block (or beyond) within the first half-time of the block. The length of
the part of the trajectory contained in the current block but not in its
sub-blocks is then at least $2 \cdot 2^{q-1}$ (to the upper boundary and
back). We may assume that such movement serves all $r$ old requests and also
the unique new request of rank $q-1$. For the purpose of the lower bound, we
may assume that this happens instantly after the block begins.
Consequently, for both sub-blocks there are no buffered old requests and the
buffer part occupied by them was freed. Therefore, in this case
\[
	T(p,q,r)\geq 2^q+2 \cdot T(p+r,q-1,0)
	\enspace.
\]
\item The strategy does not make such movement before time $2^q \cdot s + 2^{q-1}$.
Note that this is only possible for $p \geq 1$ as the unique request
of rank $q-1$ (the top left corner of the block) will be stored in the
buffer, decreasing empty buffer space to $p-1$. This request becomes old for
the later sub-block, but for the earlier sub-block there are no old requests.
Therefore, in this case we get
\[ 
  T(p,q,r) \geq T(p-1,q-1,0)+T(p-1,q-1,r+1)
	\enspace.
\]
\end{enumerate}
Combining these two cases yields the lemma.
\end{proof}

Note that the cost of $\OPT(\ell)$ in a single phase is at least
$T(\ell,\ell,0)$. In this case, it is always possible to take the second case
of the recurrence relation above, which corresponds to the basic trajectory
behavior. $T(\ell,\ell,0)$ then expands into a sum of $2^{\ell-1}$ terms of
the form $T(1,1,r)$. In our reasoning we only need $T(1,1,r) \geq 1$ but a
more careful argument could double this amount, resulting in an~(otherwise
trivial) lower bound of $2^\ell$ for the cost of $\OPT(\ell)$.

\begin{lemma}
\label{lem:Tboundimp}
	Fix any $\eta \in (0,1)$.
	Let $a = (1+\eta) \cdot \log_2 (1+1/\eta)$,
	$b_i = 2(2^i-1) \cdot \eta$ for all $i \geq 0$, and 
  \[
  \tau(p,q,r) = \frac{2^q}{a} \cdot \left(q-(1+\eta)p-b_r\right)
  \enspace.
  \]
  Then, $T(p,q,r) \geq \tau(p,q,r)$ for any $p \geq 0$, $q \geq 1$ and $r \geq 0$.
\end{lemma}

\begin{proof}
The lemma follows by induction on $q$. Its basis corresponds to the 
case $q = 1$, where 
\[
    T(p,1,r) \geq 1 \geq 2/a = \tau(0,1,0) \geq \tau(p,1,r) 
    \enspace.
\]
The second inequality follows as $a > 2$ for any $\eta \in (0,1)$, whereas the last one follows 
from the fact that $\tau$ decreases in both $p$ and $r$. 

To show the inductive step, 
by \lref[Lemma]{lem:T_relation}, it suffices to prove the 
following two inequalities (for $q \geq 2$ and $r \geq 0$):
\begin{align}
	2^q+2\tau(p+r,q-1,0) &\geq \tau(p,q,r) & \text{for $p \geq 0$} \enspace,
	\label{eq:ind-step-1} \\
	\tau(p-1,q-1,0)+\tau(p-1,q-1,r+1) &\geq \tau(p,q,r) & \text{for $p \geq 1$} \enspace.
	\label{eq:ind-step-2}
\end{align}

For \eqref{eq:ind-step-1}, we first determine a bound on 
$f(r) = (1+\eta) \cdot r - b_r + 1$. By taking its derivative in $r$ 
(i.e., $1+\eta - \eta \cdot 2^{r+1} / \log_2 \mathrm{e}$), we observe that $f$ is 
maximized at $r_0 = \log_2 (1+1/\eta) + \log_2 \log_2 \mathrm{e} - 1$,
and therefore 
\[
	f(r) \leq f(r_0) = 
	(1+\eta) \cdot ( \log_2(1+1/\eta) + \log_2 \log_2 \mathrm{e} - \log_2 \mathrm{e} - 1) 
	+ 2 \eta + 1 \leq (1+\eta) \log_2 (1+1/\eta) = a 
	\enspace.
\]
Consequently, applying $b_0 = 0$ as well, we obtain
\[
	2^q+2\tau(p+r,q-1,0) 
	= \frac{2^q}{a}\left(a+q-1-(1+\eta)(p+r)-b_0 \right) 
	\geq \frac{2^q}{a}\left(q-(1+\eta)p-b_r \right) 
	= \tau(p,q,r) 
	\enspace.
\]

For \eqref{eq:ind-step-2}, from the definition of $b_i$
it holds that $(b_0+b_{r+1})/2 = b_{r+1}/ 2 = (2^{r+1}-1) \cdot \eta = b_r + \eta$,
which implies 
\begin{align*}
	\tau(p-1,q-1,0)+\tau(p-1,q-1,r+1)
	&= \frac{2^q}{a}\left(q-1-(1+\eta)(p-1)-\frac{b_0+b_{r+1}}{2} \right) \\
	&= \frac{2^q}{a}\left(q-(1+\eta)p+\eta-\frac{b_0+b_{r+1}}{2} \right) \\
	&= \frac{2^q}{a}\left(q-(1+\eta)p-b_r \right) = \tau(p,q,r) \enspace.	
\end{align*}
\end{proof}

We now show how to prove the key lemma using \lref[Lemma]{lem:OPTcost}
and applying \lref[Lemma]{lem:Tboundimp} to $\OPT(\ell')$.

\begin{proof}[Proof of Lemma~\ref{lem:main_lemma}]
The whole phase is the block of rank $\ell$, with no buffered 
old requests and hence, by the definition of $T$, 
$T(\ell', \ell, 0) = T((1-\eps) \cdot \ell, \ell, 0)$ is a lower bound for 
the cost of $\OPT(\ell')$ in a~single phase.
Using \lref[Lemma]{lem:Tboundimp} with $\eta = \eps$, we obtain
\begin{equation*} 
  T((1-\eps) \cdot \ell, \ell, 0) \geq 
  \frac{2^\ell \cdot \left(\ell-(1+\eps)\cdot (1-\eps)\cdot \ell \right) }{ (1+\eps)\log_2 (1+1/\eps) }
  = 2^\ell \cdot \ell \cdot \frac{\eps^2}{ (1+\eps)\log_2 (1+1/\eps) }
  \enspace.
\end{equation*}
On the other hand, the cost of $\OPT(\ell)$ is at most $2^\ell$ by 
\lref[Lemma]{lem:OPTcost}, and hence the lemma follows. 
\end{proof}


\section{Cost separation for arbitrary buffer sizes}

First, we show that \lref[Lemma]{lem:main_lemma} may be ``scaled up'' appropriately,
and then we extend it to show \lref[Theorem]{thm:main_theorem}.

\begin{lemma}
\label{lem:upscaled}
Fix positive integers $\ell$ and $\beta$, and a constant $\eps \in (0,1)$.
Assume that the line consists of at least $2^\ell+1$ equidistant sites. There
exists an input sequence on which the ratio between the 
costs of two optimal algorithms, one with a
buffer of at least $\beta \cdot \ell$ and one with a buffer of at most $\beta
\cdot (1-\eps) \cdot \ell$, is at least~$\Omega(\ell)$.
\end{lemma}

\begin{proof}
First, we fix the input sequence whose existence is asserted by
\lref[Lemma]{lem:main_lemma}. In this input, we replace each request by a
\emph{packet} consisting of $\beta$ requests at the same site. 
In any optimal solution, without loss of generality, all requests belonging to
such a packet are processed together. Therefore, what matters is how many
packets can be kept in the buffers of both algorithms, i.e., at least $\ell$ in case of the
first algorithm and at most $(1-\eps) \cdot \ell$ in case of the second one. 
Hence, for the assumed buffer capacities the cost separation is 
at least $\ell \cdot \eps^2 \cdot (1+\eps)^{-1} \cdot \log_2^{-1} (1+1/\eps)$,
which is $\Omega(\ell)$ for a~fixed $\eps$.

Second, the result of \lref[Lemma]{lem:main_lemma} is still valid when 
the total number of line sites is not equal but greater than $2^\ell + 1$.
The requests arrive then only at the first $2^\ell + 1$ sites and the remaining ones 
cannot help nor hinder the performance of any algorithm.
\end{proof}

\begin{proof}[Proof of Theorem~\ref{thm:main_theorem}]
Let $m = \lfloor \log_2 (n-1) \rfloor$, i.e., $m$ is the largest integer such that $2^m+1 \leq n$.
We consider two cases, and in each of them we lower-bound the cost 
ratio between $\OPT(k)$ and $\OPT((1-\delta) \cdot k)$.
\begin{enumerate}

\item 
If $k < m$, let $\ell = k$, $\beta = 1$ and $\eps = \delta$.
\lref[Lemma]{lem:upscaled} now implies that 
the cost ratio is $\Omega(\ell) = \Omega(k)$.

\item 
If $k \geq m$, we may assume $k \geq 4/\delta$, as otherwise 
$\Omega(\min\{k, \log n\}) = \Omega(1)$ and the theorem follows trivially.
Let $\ell = \lceil m \cdot \delta / 4 \rceil$, $\beta = \lfloor k / \ell \rfloor$ and $\eps = \delta/2$,
so that $\beta \cdot \ell \leq k$. 
On the other hand, $\ell < k \cdot \delta / 4 + 1 \leq k \cdot \delta / 2$,
hence $\beta \cdot \ell > (k / \ell - 1) \cdot \ell = k - \ell
> k \cdot (1 - \delta/2)$, which finally yields
$(1-\eps) \cdot \beta \cdot \ell > (1 - \delta/2)^2 \cdot k > (1 - \delta) \cdot k$.
Therefore $\OPT(k)$ and $\OPT((1-\delta) \cdot k)$ fulfill the conditions of \lref[Lemma]{lem:upscaled},
which now implies that the cost ratio is $\Omega(\ell) = \Omega(m)$.
\end{enumerate}

In both cases we obtain that the cost ratio is 
$\Omega(\min \{k, m\}) = \Omega(\min \{k, \log n\})$.
\end{proof}


\bibliographystyle{alpha}
\bibliography{references}

\end{document}